\documentclass[runningheads]{llncs}

\usepackage[numbers, sectionbib]{natbib}
\usepackage[utf8]{inputenc}
\usepackage{amsmath}

\usepackage{algorithm,algpseudocode}
\usepackage{color}
\usepackage{amsfonts}
\usepackage{amssymb}
\usepackage{xspace}
\usepackage{comment}
\usepackage{hyperref}
\usepackage{float}

\usepackage{graphicx}

\usepackage{subfigure}
\usepackage{times}
\usepackage{sidecap}
\usepackage{enumerate}
\usepackage{tcolorbox}

\usepackage{blindtext}
\usepackage{enumitem}

\usepackage{xcolor}

\usepackage{tikzit}

\usepackage{mathtools}

\usepackage[symbol]{footmisc}

\usepackage{amsmath,amssymb, amsthm, enumerate,boxedminipage}
\usepackage{enumitem}

\def \polylog{\operatorname{polylog}}

\theoremstyle{definition}
\newtheorem*{def*}{Definition}
\theoremstyle{definition}

\theoremstyle{definition}
\newtheorem*{prf*}{Proof}
  
\theoremstyle{definition}
\newtheorem*{prfthm*}{Proof of Theorem}

\theoremstyle{theorem}

\newcommand{\etal}{et al. }

\usepackage{nameref}
\usepackage{varioref}

\newboolean{short}
\setboolean{short}{true}

\newcommand{\shortOnly}[1]{\ifthenelse{\boolean{short}}{#1}{}}
\newcommand{\onlyShort}[1]{\ifthenelse{\boolean{short}}{#1}{}}
\newcommand{\longOnly}[1]{\ifthenelse{\boolean{short}}{}{#1}}
\newcommand{\onlyLong}[1]{\ifthenelse{\boolean{short}}{}{#1}}
\newcommand{\shortLong}[2]{\ifthenelse{\boolean{short}}{#2}{#1}}
\newcommand{\longShort}[2]{\ifthenelse{\boolean{short}}{#2}{#1}} 

\onlyShort{
\usepackage{enumitem}
\setitemize{noitemsep,topsep=0pt,parsep=0pt,partopsep=0pt}
\setenumerate{noitemsep,topsep=0pt,parsep=0pt,partopsep=0pt}
}

\begin{document}
%
\title{Run for Cover: Dominating Set via Mobile Agents}
%
%
\author{Prabhat Kumar Chand\inst{}\orcidID{0000-0001-6190-4909} \and
Anisur Rahaman Molla\inst{}\orcidID{0000-0002-1537-3462} \and
Sumathi Sivasubramaniam \inst{}\orcidID{0000-0003-3605-6498}}
\authorrunning{P. Chand et al.}
%
 \institute{Indian Statistical Institute, Kolkata \\
 \email{\{pchand744, anisurpm, sumathivel89\}@gmail.com}}
\maketitle              
\vspace{-0.5cm}
\begin{abstract}
Research involving computing with mobile agents is a fast-growing field, given the advancement of technology in automated systems, e.g., robots, drones, self-driving cars, etc. Therefore, it is pressing to focus on solving classical network problems using mobile agents. In this paper, we study one such problem-- finding small dominating sets of a graph $G$ using mobile agents. Dominating set is interesting in the field of mobile agents as it opens up a way for solving various robotic problems, e.g., guarding, covering, facility location, transport routing, etc. In this paper, we first present two algorithms for computing a {\em minimal dominating set}: (i) an $O(m)$ time algorithm if the robots start from a single node (i.e., gathered initially), (ii) an $O(\ell\Delta\log(\lambda)+n\ell+m)$ time algorithm, if the robots start from multiple nodes (i.e., positioned arbitrarily), where $m$ is the number of edges and $\Delta$ is the maximum degree of $G$, $\ell$ is the number of clusters of the robot initially and $\lambda$ is the maximum ID-length of the robots. Then we present a $\ln (\Delta)$ approximation algorithm for the {\em minimum} dominating set which takes $O(n\Delta\log (\lambda))$ rounds.  

\end{abstract}

\section{Introduction}
Research on autonomous mobile agents (we interchangeably use the terms \textbf{{\em agent}} and \textbf{{\em robot}} throughout the paper) has become an area of significant interest in the field of distributed computing. As autonomous agents become part of everyday life (in the form of robots, drones, self-driving cars, etc.,) the area becomes more and more relevant. On the other hand, the dominating set problem is a well-researched classical graph problem. A dominating set $D$ of a graph $G=(V,E)$ is a subset of the nodes $V$ such that for any $v \notin D$, $v$ has a neighbour in $D$. Finding a dominating set has several practical applications. For example, in wireless communications, the dominating set of the underlying network is useful in finding efficient routes within ad-hoc mobile networks.  They are also useful in problems such as \textit{document summarising}, and for designing \textit{secure systems} for \textit{electrical grids}, to mention a few. 

Covering problems, such as vertex covers, maximal independent sets (MIS), and dominating sets, also have real-world applications in the field of mobile agents. For example, a minimum dominating set can form the charge stations or parking places for autonomous mobile robots. The maximal independent sets can be used to solve the same but also ensure that any two robots are not close in real life. Both MIS and dominating sets can find a place in the guarding problem, i.e., placing mobile robots such that robots guard (cover) an entire polygon. In the mobile robot setting, MIS has been explored in~\cite{das2022maximum,pramanick2023filling,kamei2020asynchronous}. However, in these works, the robots have some amount of visibility (they have some knowledge of the graph) while ours have zero knowledge. While both MIS and dominating sets are of great interest, in this paper, we limit ourselves to solving the dominating set problem in the field of mobile robots. We note, however, that for the minimal dominating set case, our produced dominating sets are also maximal independent sets.

While the problem is a classic problem in graph theory, several attempts have been made to solve it in distributed computing research as well~\cite{JRS02,KW03,KTW04,JKY19}.  
In the field of mobile robots, we believe ours is the first attempt to solve the dominating set problem. Our algorithms rely on the ideas used for \emph{dispersion} to achieve a solution for the dominating set problem. Dispersion, first introduced by the authors in~\cite{dis01} has been well studied for various configurations of robots (see ~\cite{dis02,dis03,dis04_cluster,dis08_global}). The problem is briefly: on an $n$ node graph $G$, in which there is a configuration of robots $\mathcal{R}, |\mathcal{R}|\leq n$, we want to ensure that there is at most one robot on each node. We take advantage of the fact that most dispersion protocols involve the exploration of $G$, and develop algorithms for calculating dominating sets. In the next subsection, we mention our major results.
\subsection{Our Results:}
In this paper, we show how to compute a minimal dominating set and an approximate minimum dominating set on an anonymous graph using mobile robots. Let $G$ be a connected and anonymous graph of $n$ nodes and $m$ edges with maximum degree $\Delta$. Suppose $n$ robots (with distinct IDs) are distributed arbitrarily over the nodes of $G$. We develop algorithms for the robots to work collaboratively and identify small dominating sets of $G$. In particular, our results are:
\begin{itemize}
    \item an $O(m)$ time algorithm for the robots to compute a minimal dominating set on $G$ when all robots are gathered at a single node initially. The set of dominating nodes also forms a MIS. 
    \item an $O(\ell\Delta\log(\lambda)+n\ell+m)$ time algorithm for the robots to compute a minimal dominating set on $G$ when the robots are placed arbitrarily at $\ell$ nodes initially. 
    \item an $O(n\Delta \log (\lambda))$ time algorithm for a $\ln (\Delta)$ approximation solution to the minimum dominating set on $G$, where $\lambda$ is the maximum ID-length of the robots.
\end{itemize}
All our algorithms require that the robots have at most $O(\log (n))$ bits of memory. In a recent work in~\cite{d2d23}, the authors solved a related problem called the Distance-2-Dispersion problem (See Section\ref{related_works} for details) which also produces a MIS in special cases (when the number of robots is greater than the number of nodes in the graph) in $O(m\Delta)$ rounds with $O(\log(\Delta))$ bits of memory in each robot.

\subsection{Related Works}\label{related_works}
Finding small dominating sets for graphs is one of the most fundamental problems in graph theory which along with its variants has been extensively studied for the last four decades. The dominating set problem and the set cover problem are closely related and the former one can be regarded as a special case of the latter. To find a minimum set cover for arbitrary graphs has been shown to be NP-hard \cite{Intractability, Karp72}. 

There have been a few studies on the dominating set and related problems for various distributed models. In~\cite{JRS02}, Jia \etal gave the first efficient distributed implementation of the dominating set problem in the CONGEST model. Their randomised algorithm, which is the refinement of the greedy algorithm adapted from~\cite{LH00} takes $O(\log (n) \log(\Delta))$ rounds and provides a $\ln(\Delta)$ -  optimal dominating set in expectation ($\Delta$ is the maximum degree among the $n$ nodes of the graph). It has at most a constant number of message exchanges between any two nodes. In ~\cite{SSR10}, Evan Sultanik \etal gave a distributed algorithm for solving a variant of the art gallery problem equivalent to finding a minimal dominating set of the minimal visibility graphs. Their algorithm 
runs in a number of rounds on the order of the diameter of the graph producing solutions that are within a constant factor of optimal with high probability. In~\cite{KW03}, Kuhn and Wattenhofer gave approximation algorithms for minimum dominating sets using LP relaxation techniques. Their algorithm computes an expected dominating set of size $(k\Delta^{\frac{2}{k}}\log(\Delta))$ times the optimal and takes $O(k^2)$ rounds for any arbitrary parameter $k$. Each node sends $O(k^2\Delta)$ messages, each of size $O(\log(\Delta))$. With $k$ chosen as some constant, this algorithm provides the first non-trivial approximation to the minimal dominating set which runs in a constant number of rounds. Fabian Kuhn \etal in ~\cite{KTW04} gave time lower bounds for finding minimum vertex cover and minimum dominating set in the context of the classical \textit{message passing} distributed model. They showed that the number of rounds required in order to achieve a constant or even only a poly-logarithmic approximation ratio is at least $\Omega(\sqrt{\frac{\log (n)}{\log (\log (n))}})$ and $\Omega(\sqrt{\frac{\log (\Delta)}{\log (\log (\Delta))}})$. In \cite{JKY19} the authors gave deterministic approximation algorithms for the minimum dominating set problem in the CONGEST model with an almost optimal approximation guarantee. They gave two algorithms with an approximation ratio of $(1+\epsilon)(1+\log(\Delta+1))$ running respectively in $O(2^{O(\sqrt{\log (n)\log(\log (n)})})$ and $O(\Delta \polylog(\Delta)+\polylog(\Delta)\log^\star (n))$ for $\epsilon>\frac{1}{\polylog(\Delta)}$. The paper also explores the problem of connected dominating sets using these algorithms, giving a $\ln(\Delta)$ - optimal connected dominating set. 

In our paper, we compute minimal dominating set and approximate minimum dominating set on an arbitrary anonymous graph with the help of mobile robots with limited memory. Our algorithms use a key procedure called \textit{dispersion} that re-positions $k\leq n$ mobile robots (which initially were present arbitrarily on a $n$-node graph $G$) to a distinct node of $G$ such that one node has at most one robot. Dispersion of mobile robots is a well-studied problem in distributed robotics in different settings. It is introduced in ~\cite{dis01} and saw its development over the years through several papers~\cite{dis02,dis03,dis04_cluster,dis08_global,dis06_byz1,dis07_byz2,dis12_fault, dis_fault}. 
To date, \cite{KS21} provides the best-known result for solving dispersion in $O(m+k\Delta)$ rounds and using $\log(k+\Delta)$ bits memory per robot ($k$ being the number of robots). 

In a recent work by Kaur \etal~\cite{d2d23}, the authors formulated and solved the Distance-2-Dispersion (\textsc{D-2-D}) problem in the context of mobile robots on an anonymous graph $G$, which is a closely related problem to ours. In the Distance-2-Dispersion problem, each of the $k$ robot settles at some node satisfying these two conditions: (i) two robots cannot settle at adjacent nodes (ii) a robot can only settle at the node already occupied by another robot if and only if there's no more unoccupied node that satisfies condition (i). They showed that, with $O(\log(\Delta))$ bits of memory per robot, the (\textsc{D-2-D}) problem can be solved in $O(m\Delta)$ rounds, $\Delta$ being the highest degree of the graph. Their algorithm requires no pre-requisite knowledge of the parameters $m,n$ and $\Delta$. Additionally, they show that if the number of robots $k\geq n$ (number of vertices in $G$), the nodes with settled robots form a maximal independent set for $G$.

\section{Model and Problem Definition}\label{sec:model}

\textbf{Graph: } The underlying graph $G(V,E)$ is connected, undirected, unweighted and anonymous with $|V| = n$ nodes and $|E| = m$ edges. The nodes of $G$ do not have any distinguishing identifiers or labels. 
The nodes do not possess any memory and hence cannot store any information. The degree of a node $v\in V$ is denoted by $\delta_v$ and the maximum degree of $G$ is $\Delta$. Edges incident on $v$  are locally labelled using port numbers in the range $[1,\delta_v]$. A single edge connecting two nodes receives independent port numbering at the two end. The edges of the graph serve as \emph{routes} through which the robots can commute. Any number of robots can travel through an edge at any given time. \\\\
\textbf{Robots: }We have a collection of $n$ robots $\mathcal{R} = \{r_1,r_2,...,r_n\}$  residing on the nodes of the graph. Each robot has a unique ID in the range $[0,n^c]$, where $c\geq 1$ is arbitrary; and has $O(\log (n))$ bits to store information. Two or more robots can be present (\emph{co-located)} at a node or pass through an edge in $G$. However, a robot is not allowed to stay on an edge. A robot can recognise the port number through which it has entered and exited a node.\\

\noindent \textbf{Communication Model:} We consider a synchronous system where the robot are synchronised to a common clock. We consider the local communication model where only co-located robots (i.e., robots at the same node) can communicate among themselves.\\

\noindent\textbf{Time Cycle: } Each robot $r_i$, on activation, performs a $Communicate-Compute-Move$ $(CCM)$ cycle as follows.
\begin{itemize}
    \item[-] \textbf{Communicate:} $r_i$ may communicate with other robots present at the same node as itself. 
    \item[-] \textbf{Compute:} Based on the gathered information and subsequent computations, $r_i$ may perform all manner of computations within the bounds of its memory. 
    \item[-] \textbf{Move:} $r_i$ may move to a neighbouring node using the computed exit port. 
\end{itemize}

\noindent A robot can perform the $CCM$ task in one time unit, called {\em round}. The \textbf{time complexity} of an algorithm is the number of rounds required to achieve the goal. The \textbf{memory complexity} is the number of bits required by each robot to execute the algorithm.   

We now give our problem definition.
\begin{definition}[Problem Definition]
Consider an undirected, connected $n$-node simple anonymous graph $G$ with $n$ mobile robots placed over the nodes of $G$ arbitrarily. Let $\Delta$ denote the highest degree of a node in $G$.\\

\noindent \textbf{Minimal Dominating Set.} The robots, irrespective of their initial placement, rearrange and colour themselves in such a way that
i) there is a robot at each node of $G$ 
ii) there is a self-identified subset of robots, coloured black, on $D\subseteq G$ which forms a minimal dominating set for $G$.\\
 
\noindent \textbf{Approximate Minimum Dominating Set.} The robots, irrespective of how they are initially placed, rearrange and colour themselves in such a way that a dominating set for $G$ of size at most $\mathbf{\alpha} |D^*|$ is identified, where $D^*$ is a minimum dominating set. Here $\mathbf{\alpha}$ is the approximation ratio to the minimum dominating set.\\ 

\noindent Our goal is to design algorithm for solving the above problems as fast as possible and keeping $\alpha$ as small as $\ln (\Delta)$. 
\end{definition}

\section{Preliminaries: DFS Traversal and \textsc{Procedure$\_$MYN}}\label{prelims}
In this section, we present two subroutines that we use in our algorithms. The first one is a simple depth-first search (DFS) traversal that allows the robots to explore the entire graph and disperse at the same time. The second one is a procedure that allows two robots on neighbouring nodes to meet each other if required.
 
\subsection{DFS Traversal Protocol}\label{DFS}
\noindent  We consider a collection of $n$ robots, $\mathcal{R}=\{r_1,r_2,\dots,r_n\}$ placed initially at a single node called the $root$.  For simplicity, we assume that $r_1$ is the lowest ID robot and $r_n$, the highest. The objective of the DFS Traversal Protocol is to explore the graph in a DFS manner until all the robots are dispersed. In addition, since $|\mathcal{R}|=n$, it is ensured that  $G$ has a distinct robot stationed at each of its nodes after the end of the protocol. We recall that, for a node $w$, the ports are numbered from $1$ to $\delta_w$, where $\delta_w$ is the degree of $w$. 

To execute DFS, each robot $r$ is provisioned with the following variables: 
\begin{itemize}
    \item $parent - $ stores the port number through which $r$ has entered an empty node and has settled (for a settled robot). 
    \item $child - $ for an unsettled robot $r$ it stores the port number last taken while entering/exiting a node. $r$, when settled stores the port number that the other robots used for exiting a node except when they entered the node in a forward mode for a second or subsequent time. $child$ is set to $0$ initially.
    \item $state - $ to indicate if $r$ is in a \textit{forward} mode or \textit{backtrack} mode. $state$ is initially set to \textit{forward}
    \item $treelabel - $ to differentiate between different DFSs arising from different clusters (meaningful only in arbitrary initial robot configuration). $treelabel$ is the ID of the smallest robot in a specific cluster. 
    \item $settled - $ to indicate whether a robot is settled ($1$) or unsettled ($0$)
\end{itemize}
\pagebreak

\noindent\textbf{\\Execution (Update Procedure)\\\\}In the first round, the robot $r_1$ assigns $r_1.settled\leftarrow1$. The remaining robots $\mathcal{R}\setminus \{r_1\}$ decide on the minimum port number available at $root$ (which is $1)$,  set $r_1.child$ and exit through port $r_1.child$ to a new empty node $u$. In the next round, the robot $r_2$ settles at $u$ and the remaining robots similarly exit through $r_2.child$ leaving $r_2$ at $u$. Let, at any stage of the DFS, the robots $\{r_{i+1},r_{i+2},\dots,r_n\}$ arrive at a node $w$ through port $p_w$ with degree $\delta_w$. Now, the robots decide on the next course of action based on its $state$ variable:
\begin{itemize}
    \item \textbf{$state={forward}$}: If $w$ is empty, set $r_{i+1}.settled\leftarrow 1$, $r_{i+1}.parent=r_{i+1}.child$. The remaining robots $\{r_{i+2},r_{i+3},\dots,r_n\}$ set $child\leftarrow(child+1)\mod{\delta_w}$ and then $r_{i+1}.child\leftarrow child$. If $child=r_i.parent$, it implies that all ports at $w$ have been used and hence the remaining unsettled robots set $state\leftarrow \textit{backtrack}$. Otherwise if $w$ is non-empty, $\{r_{i+1},r_{i+2},\dots,r_n\}$ set their $state$ to $\textit{backtrack}$. In both cases, the unsettled robots now move out through $child$. 

    \item \textbf{$state={backtrack}$}: Let $r_j$ be the settled robot at $w$. The robots $\{r_{i+1},r_{i+2},\dots,r_n\}$ set $child\leftarrow (child+1)\mod{\delta_w}$ and $r_j.child\leftarrow child$ (this updates the $child$ port on the settled robot). If $child\neq r_j.parent$ (implying an available port at $w$), the robots $\{r_{i+1},r_{i+2},\dots,r_n\}$ switch their $state$ to $\textit{forward}$ and exit through $child$ port.
\end{itemize}
The protocol ends when there are no more unsettled robots remaining. Out of the $m$ edges, each $m-(n-1)$ non-tree edge is traversed a maximum of four times (twice from either end) and the tree edges are traversed twice. Hence, the maximum number of rounds required to execute the DFS Traversal Protocol is $4(m-n+1)+2(n-1)=4m+2n-2$. So, we have the following lemma.
\begin{lemma}\label{DFS_Lemma}
    Let $G$ be a $n$-node arbitrary, connected and anonymous graph with a maximum degree $\Delta$. Let $n$ mobile robots with distinct IDs be placed initially at a single node. Then, the DFS Traversal Protocol disperses each of the mobile robots into $n$ distinct nodes in $O(m)$ rounds. 
\end{lemma}


    


\subsection{\textsc{Procedure$\_$MYN} ({\sc Meet-Your-Neighbor})}
Since the port ordering of nodes is different, it can be tricky to ensure that two robots on neighbouring nodes meet. Also, it is difficult to time the movement of robots and guarantee that two neighbours meet without access to a global clock. Since, it can be essential for two robots to pass information to each other, arranging ways to ensure such a meeting can be beneficial.  \textsc{Procedure$\_$MYN} is helpful in ensuring that a robot is able to communicate with all its neighbours at least once when required.  We use the pairing procedure \textsc{Procedure$\_$MYN} to ensure that during a scan for neighbours, a robot meets all its neighbours. For this, the algorithm essentially exploits the bits representing the IDs of the robots. Let $\lambda$ denote the largest ID among all the $n$ robots. Therefore, the robots use a $\log (\lambda)$ bit field to store the IDs. \textsc{Procedure$\_$MYN} runs in phases. Each phase consists of $\Delta$ rounds and there are a total of  $\log (\lambda)$ phases. Each phase corresponds to a bit in the field (with robots having IDs less than $\log (\lambda)$ bits padding the rest with $0$s). The steps in a phase are simple, starting with the rightmost bit, if the bit is $1$, the robot uses $\Delta$ rounds to visit all its neighbours. If the bit is $0$, the robot waits at its node for visitors. Clearly, since all robots have unique IDs, for any two pairs of neighbouring robots, there exists at least one round in which the robots have different bits and meet. A detailed procedure of the MYN procedure can be found in Algorithm~\ref{alg: procedure_P}. 
Now, to find the neighbouring robots, a robot $r_i$ stationed at a node $u$ does the following.
\begin{enumerate}\label{procedure1}
    \item Checks the rightmost bit in its ID field. 
    \item If the bit is $1$, $r_i$ goes to each of its neighbours one by one from $u$ following ascending order of the port numbers and back. Note that, from $u$, there are ports numbered from $1$ to $\delta_u$ each leading to a distinct neighbour of $u$.
    \item If $r_i$ finds a robot $r_j$ (or other multiple robots) in its neighbour, it can exchange required information.
    \item Otherwise, if the bit id $0$, the robots sits at $u$ and waits for $\Delta$ rounds. 
    \item In the next phase of $\Delta$ rounds, $r_i$ now checks its second bit from the right and repeats the same process as described in the previous steps.   
    \item The process continues for $\log \lambda$ phases, ensuring every bit in the ID field has been scanned. If there are no more bits remaining in $r_i.ID$ and $\log\lambda$ rounds have not been completed (implying that $r_i$ has a smaller ID length than $\log\lambda$ bits), $r_i$ assumes the current bit as $0$ and stays back at its own node for the rest part of the algorithm. 
\end{enumerate}
Clearly, the procedure takes no more than $O(\log (\lambda))$ rounds to ensure that two specific neighbours meet, thus to ensure that a robot meets with all its neighbours it takes no more than $O(\Delta\log (\lambda))$ rounds. Hence,
\begin{lemma}\label{lem: pairing}
 \textsc{Procedure$\_$MYN} ensures that a robot meets all its neighbours at least once and takes no more than $O(\Delta\log (\lambda))$ rounds.    
\end{lemma}

\begin{algorithm}[ht!]
    \caption{\textsc{Procedure$\_$MYN}}
    \begin{algorithmic}[1]
    \Require{A mobile robot $r$ with $\log\lambda$ bit ID field - $b_1b_2...b_{\log \lambda}$; }
    \Ensure{$r$ meets every other robot in its neighbourhood}
    \Statex
    \For{$i=\log\lambda$ to $1$}         
        \If{$b_i$ is $1$}
            \For{$j=1$ to $\Delta$}
            \State $r$ visits neighbour at the port $\delta_j$. If the neighbour contains a robot, it can communicate with it.
            \EndFor
            
        \Else
            \State $r$ remains at its node for $\Delta$ rounds.
        \EndIf
   
    \EndFor
    \end{algorithmic}
\end{algorithm}\label{alg: procedure_P}

\section{Algorithm for Minimal Dominating Set}
In this section, we show that we can achieve a minimal dominating set for both the rooted and arbitrary initial configuration. In the rooted case, initially, all robots are gathered at a single root, while in the arbitrary case, the robots are gathered in clusters across the graph. 
\subsection{Single Source Initial Configuration} \label{sec:rooted-config}
Let us first consider the case where all the mobile agents are initially housed at a single node of the graph. We refer to such a configuration as a single source or \emph{rooted initial configuration}. We design an algorithm for the mobile agents to work collaboratively to compute a minimal dominating set of $G$. In particular,  agents identify a subset $D$ of the vertices $V$ such that $D$ is a minimal dominating set of $G$. Given $n$ agents that start from a single source, our algorithm takes $O(m)$ rounds to find such a $D$, where $m$ is the number of edges in the graph. Agents are not required to know any graph parameters.     

To solve the dominating set problem in the rooted initial configuration, we 
modify the standard DFS traversal (see Section~\ref{DFS}) to allow the agents to simultaneously traverse and compute the dominating set $D$. 

To ensure that the robots know if they are part of the dominating set for $G$ in the end, each robot is \emph{coloured} accordingly. The robots that occupy the nodes $D\subset V$  are coloured ``black'' to distinguish them from other mobile robots. Thus, the set of mobile robots coloured black forms the dominating set for $G$. To ensure proper colouring, each robot uses an additional variable $r_i.colour$ along with the ones used for executing the DFS Traversal.
$r_i.colour$ is used to classify the nature of the robot with respect to the dominating set. Each robot, at a given time, has one out of the three colours - \textit{black}, \textit{grey} and \textit{white}. Initially, all robots are \textit{white} in colour. Robots change their colour to \textit{black} once it becomes a member of the dominating set. Robots that are adjacent to a black coloured robot are coloured \textit{grey}.

We will now describe our algorithm in detail. As mentioned before, our algorithm for creating a dominating set follows a modified version of the depth-first algorithm described in Section~\ref{DFS}.
We do so by the addition of an extra step each time a robot explores a new node, to decide a robot's colour.

As in the DFS, the collection of robots $\mathcal{R}=\{r_1,r_2,\dots,r_n\}$ start the algorithm from the $root\in V$.  The robots leave the smallest ID robot $r_1$ at $root$, where $r_1$
settles and colours itself \textit{black}. In general, the robots navigate the graph via the DFS protocol with the extra step of deciding the colour of the robot that settles. This is decided as follows.
Each time the robots decide to settle a robot a the node $u$, they visit all neighbours of $u$. If they find a \textit{black} settled robot among $u$'s neighbours, then the robot settling at $u$ colours itself \textit{grey}, else it colours itself \textit{black}. However, the robot that has just arrived at a new empty node with its $parent$ coloured \textit{black} can immediately become \textit{grey}. The remaining robots move via the smallest port available at $u$ to the next node according to the DFS traversal protocol (see, Section~\ref{DFS}) and the algorithm continues. 

The unsettled robots scan their neighbour for \textit{black} robots only in the \textit{forward} phase. The algorithm stops when the last unsettled robot that started from $root$ settles and colours itself. The nodes of $G$, which we notate by $D$, that have a \textit{black} robot placed, now form a dominating set for the graph $G$.

\begin{algorithm}[ht!]
    \caption{\sc Dominating Set - Single Source (Algorithm for robot $r_i$)}
    \begin{algorithmic}[1]
    \Require{An $n$-node anonymous graph with $n$ mobile agents docked at $root$.}
    \Ensure{Agents settle over the nodes and identify a minimal dominating set.}
    \Statex
    \State Let the robots $\{r_1,r_2,\dots,r_i,\dots,r_n\}$ be docked at a node $root$. Robot $r_i$ maintains the following list of variables: $\langle r_i.id, r_i.parent, r_i.child, r_i.state, r_i.colour\rangle$, where $r_i.colour$  is initially set to \textit{white}, $r_i.parent$, $r_i.child$ are set to $null$ and $r_i.state$ is set to \textit{forward}.   
    \While{ $r_i$ is unsettled}
        \If{ $r_i$ is not the minimum ID robot among the unsettled robots}
            \State move according to DFS Traversal Protocol~\ref{DFS}
        \Else
            \If{the current node is empty}
                \State set $r_i.settle\leftarrow 1$ and inform the other unsettled robots to stay stationary
                \State move to and check the colour of its parent robot and return
                    \If{ parent is black}
                        \State sets $r_i.colour\leftarrow$\textit{grey}
                    \Else 
                        \State visits each neighbour of the current node and if there is at least one \textit{black} in its neighbour, set $r_i.colour \leftarrow$\textit{grey} or else set $r_i.colour\leftarrow$\textit{black}
                        \State inform the remaining unsettled node about the completion of colouring 
                    \EndIf
            \Else \State move according to DFS Traversal Protocol~\ref{DFS} to find an empty node
            \EndIf
        \EndIf
    \EndWhile
    \State remain stationary at the current node for the rest of the algorithm
    \State The nodes where each \textit{black}-coloured robot is stationed form the dominating set $D$ of $G$.
    \end{algorithmic}
\end{algorithm}\label{alg: rooted_ds}

\begin{figure}[!ht]
  \centering
    \includegraphics[scale=0.3]{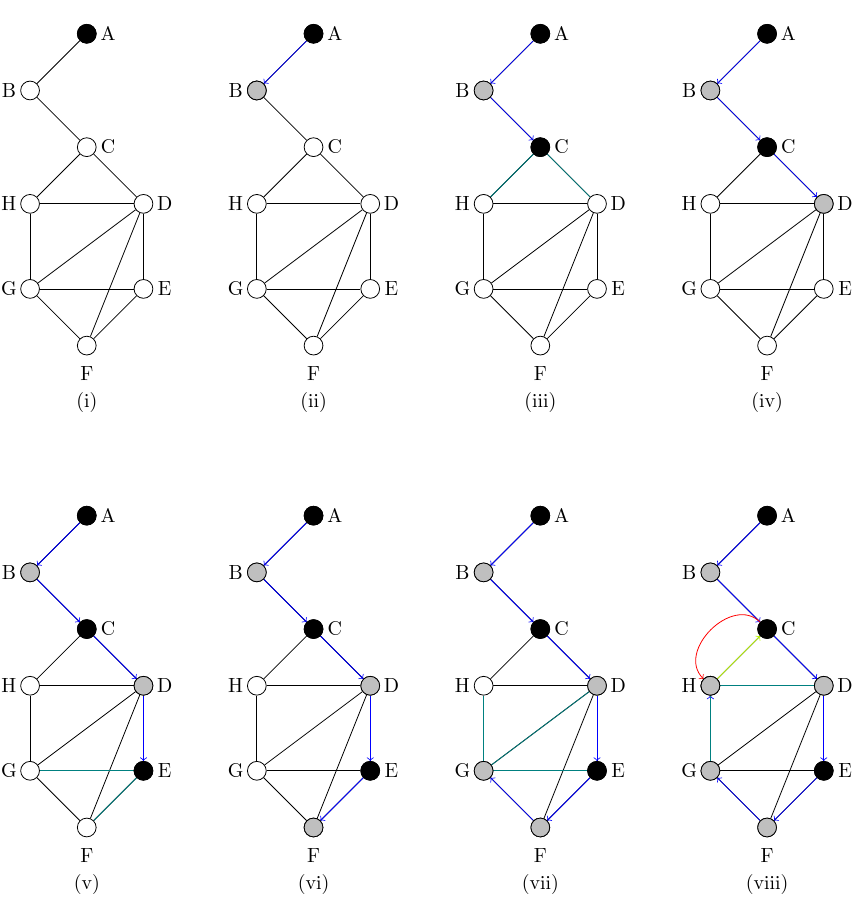}
  \caption{Illustration of Mobile Robots identifying a Minimal Dominating Set. The robots start at the $root$ node A. The minimum ID robot settles at A and colours itself as \textit{black}. The rest of the robots follow the DFS and ends up at node $B$, where the next minimum ID robot in the group settles and colours itself as \textit{grey} (since it comes from a \textit{black} parent). The rest of the robots now stops at the next node in the DFS, node C. The settled robot at C scans the neighbourhood (teal coloured edges) and finds no \textit{black} robot and hence colours itself as \textit{black}. The algorithm continues along the DFS route (blue-directed edges) shown from (iv) to (vii), settles and colours robots accordingly till it reaches the final configuration (viii). The red edge denotes a $backtrack$. After the configuration of (viii) is reached, no new node is discovered and the algorithm stops after completing the DFS of the graph.}\label{fig:rooted_mds}
  \end{figure}

\subsubsection{Analysis}

\begin{lemma}\label{root_minimal}
Algorithm~\ref{alg: rooted_ds} ensures that each robot is associated with a distinct node in $G$ and that the subset, $D$, of robots coloured black forms a minimal dominating set of $G$. 

\end{lemma}
\begin{proof}
During the course of the algorithm, which begins with a group of unsettled robots at a specially designated node called $root$ 
, the robots get settled one by one in increasing order of the IDs, at some distinct node of $G$. The DFS protocol ensures that the entire graph is explored and each node in $G$ contains a settled robot. Since each settled robot gets immediately coloured as black or grey, Algorithm~\ref{alg: rooted_ds} leaves no robot with the colour \textit{white}. The \textit{black}-coloured robots, $D$, form the dominating set, whereas the \textit{grey} robots are the ones that are adjacent to one or more robots included in the dominating set. Note that a robot is coloured grey if and only if there exists a neighbour that was coloured black. Since there are no white robots at the end of the protocol, $D$ is a dominating set of $G$.

We now prove that $D$ represents a minimal dominating set for the graph $G$. For contradiction, let us assume that after the end of the algorithm, we can remove a \textit{black} node (robot) $r_x$ from $D$ without disrupting it's dominating set property. Therefore, $D$ is still a dominating set for $G$ without $r_x$. Now, since the algorithm executes sequentially, the parent and every neighbourhood of $r_x$ is \textit{grey} after the end of the algorithm. It leaves the node $r_x$ uncovered by any black node. Therefore $D$ does not form a dominating set for $G$ with $r_x$ excluded.
\end{proof}


\begin{lemma}\label{time_rooted}
    Algorithm~\ref{alg: rooted_ds} executes in $O(m)$ rounds. 
\end{lemma} 
\begin{proof}
    The algorithm essentially executes a depth-first search, in which it settles the $n$ robots, one by one. The depth-first search takes up $O(m)$ rounds. In addition to that, once a robot $r_i$ settles at a specific node $u$, it takes a maximum of $O(\delta_i)$ rounds, where $\delta_i$ is the degree of the node $u$, to search for robots among its neighbours, to decide its own colour. Therefore, the algorithm takes $O(m+\sum(\delta_i))=O(m)$ rounds to execute.  
\end{proof}
\noindent \textbf{Memory per robot.}\label{memory_rooted}
Each robot stores its ID which takes $O(\log (n))$ bit space. Along with that, the parent and child pointers take $O(\log (\Delta))$ bit memory each. Other variables take up a constant number of bits. Therefore, the memory complexity is $O(\log(n+ \Delta))=O(\log (n))$ bits. 

Combining it with Lemma~\ref{root_minimal} and Lemma~\ref{time_rooted} and we have the following theorem. 
\begin{theorem}
     Let $G$ be an $n$-node arbitrary, connected and anonymous graph with $m$ edges. Let $n$ mobile robots with distinct IDs in the range $[0,n^c]$, where $c$ is constant, be placed at a single node, known as the rooted initial configuration. Then, a minimal dominating set for $G$ can be found in $O(m)$ rounds using Algorithm~\ref{alg: rooted_ds} with $O(\log (n))$ bits of memory at each robot.
\end{theorem}

As a by-product, the Algorithm~\ref{alg: rooted_ds} computes a maximal independent set (MIS). In fact, the set $D$ is also a MIS for $G$. Thus we get the following result on MIS.  

\begin{theorem}[Maximal Independent Set]\label{lem:mis-single-source}
Let $n$ mobile robots be initially placed at a single node of an $n$-node anonymous graph $G$ with $m$ edges. Then there is an algorithm (cf. Algorithm~\ref{alg: rooted_ds}) for the robots to compute a maximal independent set of $G$ in $O(m)$ rounds and with $O(\log (n))$ bits memory per robot.
\end{theorem}
\begin{proof}
    A robot receives a \textit{black} colour after ensuring that no neighbour is coloured black. This ensures that no two adjacent nodes host \textit{black} robots guaranteeing $D$ to be an independent set. Now, let's assume that we can add a black node $v$ to $D$ while maintaining its independent set property. Since, the robot at $v$ was initially \textit{grey}, it must have had a \textit{black} neighbour at some node $u$. Therefore, we happen to have two adjacent nodes $u$ and $v$ that are resided by \textit{black} robots, contradicting the fact that $D$ is an independent set. Therefore, no new black node can be added to the independent set $D$, making it a maximal independent set (MIS). 
\end{proof} 

\subsection{Multi Source Initial Configuration}\label{mds_arbitrary}
In the multi-source (or arbitrary) case, the $n$ robots $\{r_1,r_2,\dots,r_n\}$ start as $\ell<n$ clusters, placed arbitrarily at $\ell$ different nodes of $G$. Once again, our algorithm for constructing dominating sets is inspired by the DFS procedure for dispersion~\cite{KS21}.
In this section, we show how we can achieve a minimal dominating set for such a configuration.  

The robots first perform dispersion using the method described in Kshemkalyani\cite{KS21}. Briefly, their algorithm allows each of the $\ell$ clusters to begin DFS independently. If a DFS meets no other DFS, then it executes to completion as in the rooted case. If not, their algorithm ensures that if two (or more) DFSs meet while dispersing, they are all merged into the DFS with the largest number of settled robots. That is if DFS $i$ meets DFS $j$ and $i$ has more settled robots at the time of meeting, then $j$ is collapsed and all of $j$'s robots join in executing $i$'s DFS. This act of gathering all of $j$ into $i$ is called \emph{subsumption}. Thus, when two DFSs meet, the shorter DFS gets subsumed by the larger one. Note that this means that once dispersion is achieved, there may be several DFS trees.

Let us assume then, after the end of the algorithm, there are $\ell^* \leq \ell$ independent DFSs that never met (each with its own unique $root$). Each DFS is identified by a unique label $treelabel$ (which is the ID of the smallest robot in the tree). Our procedure to create the minimal dominating set consists of two steps i) elect a global leader robot $r^*$ ii) with $r^*$ as root, use an adaption of Algorithm~\ref{alg: rooted_ds} to achieve a minimal dominating set. Note that the dispersion algorithm in~\cite{KS21} takes $O(m)$ rounds. To ensure coordination, each robot waits to begin the leader election protocol after $cn\Delta$ rounds for a sufficiently large constant  $c$ from the start of dispersion.\\



\noindent\textbf{Leader Election:} We first use the \textsc{Procedure$\_$MYN} (see Section~\ref{prelims}) to ensure that all robots first meet their neighbours. When a robot meets a robot from a different ID, they share their $treelabel$s with one another. When a robot receives a lower $treelabel$ value, it updates its own $treelabel$ to the lower value. After \textsc{Procedure$\_$MYN} has been executed, each robot, starting from the leaf nodes in a DFS (say DFS $i$) now sends its newly updated $treelabel$ value to the root of DFS $i$ ($root_i$) via its $parent$ pointers. The root of the DFS $i$, $root_i$, waits for $O(n)$ rounds as it receives different values of $treelabel$ continuously from its leaf nodes (\textit{up-casting}). The $root_i$ now compares the minimum $treelabel$ values received from its children with its own $treelabel$ and decides on the minimum $treelabel$. It then sends it along through its children into each and every node of its, in the DFS (\textit{down-casting}). All the robots in DFS $i$ are now updated with the lowest $treelabel$ value. This marks the end of a phase. After the end of $\ell$ (although $\ell^*$ is sufficient, its value is unknown) such phases, it is guaranteed that each of the $n$ robots in $G$ now has a consistent $treelabel$ value. The final $treelabel$ value is the minimum among the $\ell^*$ DFSs, however, it may not represent the minimum $ID$ robot, as it may be consumed by a larger DFS during the dispersion process happening earlier. After the end of the protocol, the robot that has the globally decided minimum $treelabel$ continues the algorithm. We identify the leader robot with the minimum $treelabel$ value as $r^\star$.\\

\noindent\textbf{Creating Minimal Dominating Set:} With $r^\star$ as a leader, it first creates a new DFS of the graph $G$ with the node containing $r^\star$ as the $root$ node. Since the robots are dispersed across each and every node of $G$, the nodes of $G$ are now distinguishable owing to the distinct IDs of the mobile robots in the nodes. $r^\star$  starts from the root and rewrites the $parent$ and $child$ pointers of each robot as its traverses across the whole graph $G$ in a depth-first manner. Moving to an empty node implies that $r^\star$ has arrived at the $root$. In such cases, $r^\star$ modifies its own variables accordingly. The other details regarding the DFS can be found in Section~\ref{DFS}. After the end of this process, which takes an additional $O(m)$ rounds, the newly assigned pointers of the robots now represent a DFS of $G$ with $r^\star$ (node) as the root.

Initially, all the robots are coloured \textit{white} by default. The identification of the dominating set takes place in similar lines as with the rooted case described in the previous Section~\ref{alg: rooted_ds}. The algorithm begins with $r^\star$ colouring itself black. $r^\star$ then moves the next node (say, node $u$ resided by a robot $r_u$) following the newly created DFS. Upon the arrival of $r^\star$ at $u$, $r_u$ scans each of its neighbours to check if there is a black robot. If there are no black robots, $r_u$ colours itself as \textit{black} otherwise if there is at least one black robot in the neighbourhood of $u$, $r_u$ colours itself \textit{grey}. After $r_u$ has changed its colour from \textit{white}, it informs $r^\star$ as $r^\star$ now moves to the next node following the DFS. The process continues similarly with each new node of degree $O(\delta_i)$, taking $O(\delta_i)$ rounds its scan its neighbour and colour itself accordingly. The algorithm terminates as soon as $r^\star$ completes the DFS of $G$.

Therefore, combining the results from previous section (Section~\ref{sec:rooted-config}, DFS traversal and \textsc{Procedure$\_$MYN}, we  get the following main result. 
\begin{theorem}
   
        Let $G$ be an $n$-node anonymous graph with the maximum degree $\Delta$. Let $n$ mobile robots with distinct IDs in the range $[0,n^c]$, where $c$ is constant, be initially placed arbitrarily among $\ell$ nodes of $G$ in $\ell$ clusters. Then, a minimal dominating set for $G$ can be found in $O(\ell\Delta\log(\lambda)+n\ell+m)$ rounds using Algorithm~\ref{mds_arbitrary} with $O(\log (n))$ bits memory per robot. $\lambda$ is the maximum ID-string length of the robots. 
\end{theorem}
\begin{proof}
    Algorithm~\ref{mds_arbitrary} first disperses the $n$ robots over the $n$ nodes using the algorithm in~\cite{KS21} , which takes $O(m)$ rounds. We then elect a leader among the robots. There are at most $\ell$ DFS trees (as it starts with $\ell$ clusters initially). Each DFS tree uses \textsc{Procedure$\_$MYN} to communicate with the neighbouring DFS (if any). Once the leaf nodes receive a new communication from a different DFS, it sends the information to the root of its DFS; which may take $O(n)$ rounds. In the worst case
    \textsc{Procedure$\_$MYN} may be executed $\ell$ times to elect the leader. From Lemma~\ref{lem: pairing} we know that  \textsc{Procedure$\_$MYN} takes $O(\Delta \log(\lambda))$ rounds. Thus, the leader election requires $O(\ell(\Delta\log(\lambda)+n))$ rounds. The leader robot then finds a minimal dominating set while forming a single DFS tree using at most $O(m)$ rounds. Finally, the dominating set identification takes another $O(m)$ rounds. Therefore, the time complexity of Algorithm~\ref{mds_arbitrary} is $O(m + \ell(\Delta\log(\lambda)+n)+ m) =O(m+\ell\Delta\log(\lambda)+n\ell)$ rounds. 
\end{proof}

\section{Algorithm for $\ln(\Delta)$-Approximation Minimum Dominating Set}\label{sec:approx-algo-minD}
In this section, we describe an approximation algorithm that gives us a dominating set with a size that is optimal within a factor of $\ln(\Delta)$ in the mobile agent setting. In the previous algorithms, our methodology was based on scanning only the 1-hop neighbourhood of a robot to assign a robot one of the two colours - \textit{black} or \textit{grey}. The set of \textit{black} coloured robots formed the dominating set for $G$. Although these previous algorithms (\textit{rooted}) ran in $O(m)$ rounds, the approximation ratio of the dominating set size produced with these algorithms to the optimal (minimum) dominating set could be huge. For example, if we consider a star graph with the robots starting at one of the leaves, then all leaf nodes become included in the dominating set (Figure \ref{fig:non_optimal_ds}). Although that dominating set is minimal, it is far from optimal. The optimal in the case of a star graph is the single non-leaf node at the centre. Therefore, the approximation ratio in such a case can be as large as $O(n)$. In this section, we adapt the well-known greedy distributed algorithm~\cite{LH00}, to provide a better approximation solution for the dominating set problem in the mobile agent setting. 

\begin{figure}[!ht]
  \centering
    \includegraphics[scale=0.25]{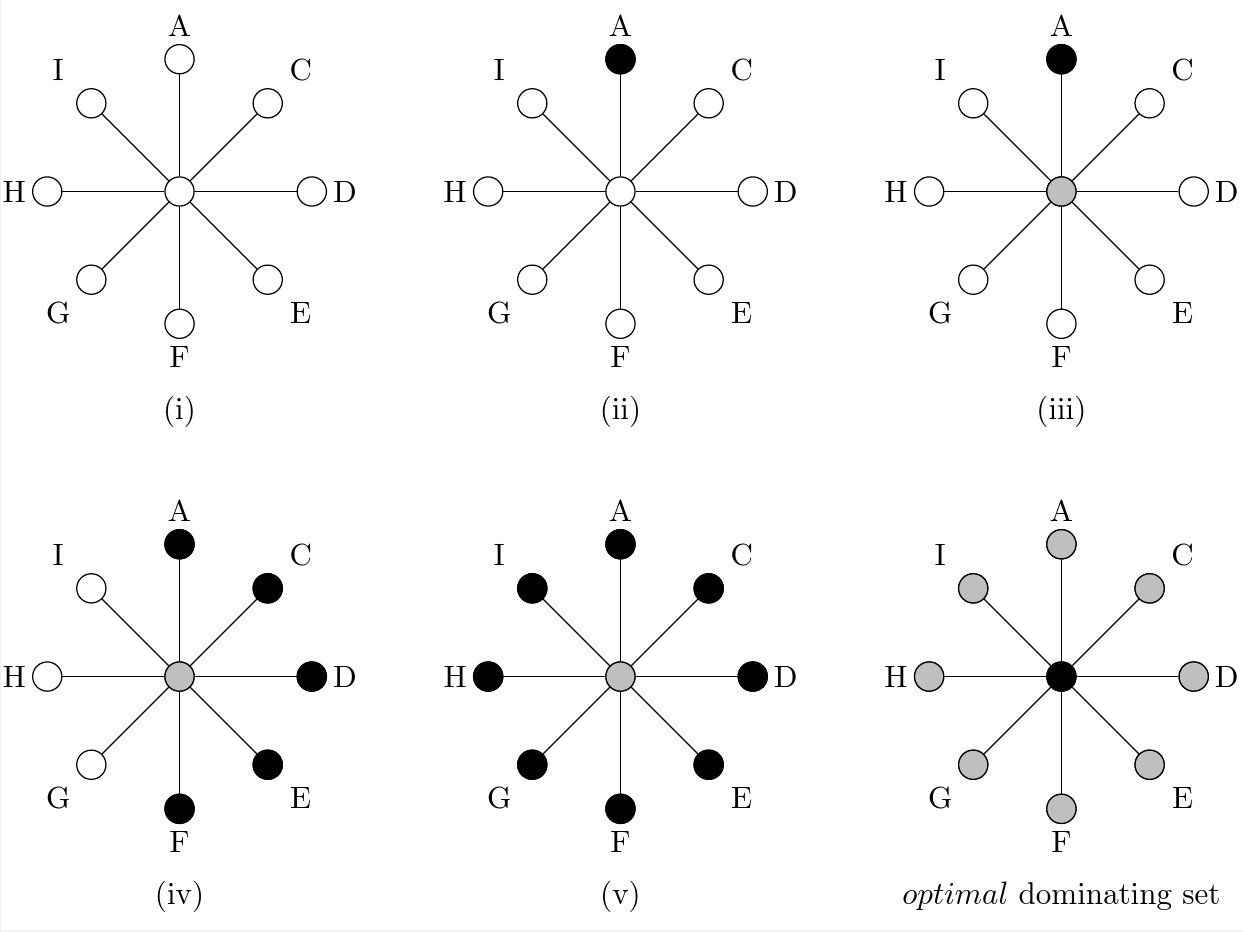}
  \caption{When the robots start at a leaf node (node A), our previous algorithm (Algorithm~\ref{alg: rooted_ds}) produces a minimal dominating set identified by the black nodes as in (v) by executing through (ii), (iii) and (iv). The figure right next to (v) on the other hand, shows an optimal dominating set. Therefore, for an $n$ node graph, the size of a minimal dominating set could be as bad as $O(n)$ compared to the optimal.} \label{fig:non_optimal_ds}
  \end{figure}

In~\cite{LH00}, the authors give a greedy distributed algorithm (referred to in the paper as ``\textit{distributed database coverage heuristic} (DDCH)") for calculating a dominating set with an approximation ratio of $\ln\Delta$ to the optimal set size. The same is reintroduced for a synchronous communication model where the nodes can communicate with each other without moving~\cite{JRS02}. The authors also refined the algorithm to a randomized version obtaining a better-expected running time. We essentially adopt the greedy distributed algorithm to our own model. The challenge is of course as follows: in the earlier distributed setting, nodes are allowed to communicate with their neighbours within the span of a round, whereas the robots in our model cannot communicate with each other unless they are at the same node. While the robots can gather information by visiting neighbouring robots, however, it is difficult to time the movement of robots and guarantee that two neighbours meet without access to a global clock. However, with the help of \textsc{Procedure$\_$MYN} introduced in the previous section, a successful adaptation of the greedy algorithm is possible. More details follow.

For this protocol, we assume that the $n$ mobile robots of set, $\mathcal{R}=\{r_1,r_2,\dots,r_n\}$ start from a dispersed position on graph $G$. That is, they are initially dispersed among the $n$ nodes of the graph $G$ with each robot at a distinct node of $G$. If the robots are arbitrarily placed over $G$ initially, the robots could be re-positioned using the dispersion algorithm described in ~\cite{KS21}. As in our previous sections, each robot is assigned a variable colour to keep track of its colour, which can take one of three values \emph{white, grey} or \emph{black}. Let $D\subset G$ be the set of black robots at the end of the protocol, and $D$ forms the dominating set of $G$. A white robot is one that is not (yet) covered by any robot in the dominating set. The \textit{grey} robots represent robots that are already covered by some robot in $D$. Initially, all the robots are coloured \textit{white}.

Let $span$ of a robot $r_i$, $w(r_i)$, be the number of white robots in the direct ($1-hop$) neighbourhood of $r_i$, including $r_i$ itself. Each robot $r_i$ uses an extra variable $r_i.span$ for the $span$ values. $r_i.span$ has both the ID and the span count. The basic idea of the protocol is as follows. Each robot calculates the maximum span value within its $2-hop$ neighbourhood. If $r$ is the robot that has the maximum span within $2-hops$ then $r$ colours itself black and informs its neighbours. To achieve this, each robot performs the following four stages (of $2\Delta \log (\lambda)$ rounds each) until there are no white robots left in the graph:
\begin{enumerate}
    \item \textbf{$r_i$ calculates its span $r_i.span$: }
    The \textsc{Procedure$\_$MYN}, described in the previous section guarantees that $r_i$ meets all its neighbouring robots within  $O(\Delta\log(\lambda))$ rounds (``meeting" includes $r_i$ being stationary and a neighbouring robot arriving at $r_i$ during these $O(\Delta\log(\lambda))$ rounds). Inside the first $O(\Delta\log(\lambda))$ rounds, the robot $r_i$ communicates with its immediate neighbours and evaluates $r_i.span$.  
    
    \item \textbf{$r_i$ gets to know the highest $span$ value within its immediate neighbours}\\
    Inside the next $O(\Delta\log(\lambda))$ rounds, $r_i$ communicates with its immediate neighbours to know the highest $span$ value within its neighbourhood. If $r_j.span$ is the highest span value among immediate neighbours, $r_i.span$ gets replaced by $r_j.span$. For identical $span$ values, a robot with a lower ID is selected. 
    \item \textbf{$r_i$ gets to know the highest $span$ value within its $2-hop$ neighbours}\\
    In the next $O(\Delta\log(\lambda))$ rounds of the algorithm, $r_i$ visits all its immediate neighbours once again to check if there are any new updated span values (possibly from a neighbour's neighbour). 
    \item \textbf{if $r_i$ coloured itself black, then it informs its neighbours.}\\
    After the completion of the first three stages, $r_i$ has the information of the highest span and the robot that has the highest span in its $2-hop$ neighbourhood. If $r_i$ itself is the robot with the highest span, it colours itself \textit{black}. Any robot that receives \textit{black} colour then takes additional $O(\Delta\log(\lambda))$ rounds to go to its neighbours and colour any \textit{white} robot in their neighbour as \textit{grey}
\end{enumerate}
Note that once a robot colours itself  \textit{black}, its colour does not change.  All robots then reset their $span$ values and the next phase of $O(\Delta\log\lambda)$ begins. The algorithm runs till each robot $r_i$ has no more \textit{white} robots in their neighbourhood.  

\begin{algorithm}[ht!]
    \caption{\sc $\ln(\Delta)$ Optimal Dominating Set (Algorithm for robot $r_i$)}
    \begin{algorithmic}[1]
    \Statex
    \While{there are white robots in the neighbourhood} 
       \For{$i=1$ to $O(\Delta \log \lambda)$}  \Comment{Stage $1$}
       \State Each robot $r_i \in \mathcal{R}$ computes its span $r_i.span$ using compute span  \textsc{Procedure$\_$MYN}.
       \If{$r_i.span$ is zero}
         \State Then $r_i$ stops executing the protocol. But will provide information if requested. 
       \EndIf
       \EndFor
             \For{$i=1$ to $O(\Delta \log \lambda)$}  \Comment{Stage $2$}
       \State Each robot $r_i \in \mathcal{R}$ communicates its span $r_i.span$ its neighbours.
       \EndFor 
             \For{$i=1$ to $O(\Delta \log \lambda)$}  \Comment{Stage $3$}
       \State Each robot $r_i \in \mathcal{R}$ computes the maximum span within 2-hops.
       \EndFor
       
        \If{$r_i.span$ is highest among all robots in the  $2-hop$ neighbourhood},  \Comment{Stage $4$}
            \State $r_i$ colours itself as \textit{black}
            \For{$i=1$ to $O(\Delta \log \lambda)$}
            \State $r_i$ meets each neighbour and colours any \textit{white} robot as \textit{grey} (use \textsc{Procedure$\_$MYN}) 
            \EndFor
        \EndIf
        \State Reset $r_i.span$ for all $r_i\in \mathcal{R}$.
    \EndWhile
    \end{algorithmic}
\end{algorithm}\label{alg: approx_DS}

\begin{lemma}\label{lemma: 2_hops}
    Let $r_i$ and $r_k$ be two robots at a distance of $2$ hops from each other. Then, the value of $r_k.span$ can be communicated to $r_i$ within $O(\Delta\log(\lambda))$ rounds.   
\end{lemma}

\begin{proof}
    Since $r_i$ and $r_k$  are at a distance of $2$ hops from each other, there exists a sequence of robots (nodes) $(r_i,r_{i_{1}},r_k )$ from $r_i$ to $r_k$. In the first $O(\Delta\log(\lambda))$ rounds, $r_i$ can communicate with $r_{i_{1}}$ to get the value of $r_{i_{1}}.span$. In the meanwhile, $r_{i_{1}}$ also collects the value of $r_k.span$ in the same $O(\Delta\log(\lambda))$ round. Therefore, in the second sub-phase of $O(\Delta\log(\lambda))$ rounds, when $r_i$ communicates with $r_{i_{1}}$ again, $r_i$ receives the value of $ r_k.span$ through $r_{i_{1}}$ or the vice-versa.
\end{proof}

\begin{lemma} \label{lem: approximation}
    Algorithm~\ref{alg: approx_DS} computes a dominating set $D$ which is a $\ln(\Delta)$ approximation to the optimal size dominating set $D^*$ for the graph $G$. 
\end{lemma}
 
\begin{proof}
Our protocol emulates the greedy distributed algorithm in~\cite{JRS02} by ensuring that the node with the highest span within a two-hop neighbourhood becomes a part of the dominating set. Hence the $\ln(\Delta)$ approximation follows directly from the approximation results in~\cite{JRS02,lec_rw}.
\end{proof}
\begin{lemma}\label{lem: execute}
     Algorithm~\ref{alg: approx_DS} takes $O(n\Delta\log(\lambda))$ rounds to execute. 
\end{lemma}
\begin{proof}
 There are four stages within a single while loop in Algorithm~\ref{alg: approx_DS}. The first stage starts by calculating the span of each robot which takes $O(\Delta\log(\lambda))$ rounds. In the next two stages of $O(\Delta\log(\lambda))$ rounds, the robots communicate their $span$ to all the robots with a distance of $2-hops$. In the final stage, when a robot gets a colour \textit{black} (the robot with the highest span in its $2-hop$ neighbourhood), it can take another $O(\Delta\log(\lambda))$ rounds to instruct its neighbouring robots to colour themselves as \textit{grey}. So, a single while loop from $span$ calculation till colouring robots as \textit{grey}; takes no more than $O(\Delta\log(\lambda))$ rounds. As the execution of a single while loop gives us at least one black robot
 , in the worst case, the algorithm needs at most $O(n)$ iterations of the while loop. Thus, giving us a complexity of $O(n\Delta\log(\lambda))$ rounds.
\end{proof}

\begin{theorem}
   Let $G$ be a $n$-node arbitrary, connected and anonymous graph with a maximum degree $\Delta$. Let $n$ mobile robots with distinct IDs are placed in a dispersed initial configuration. Then, a $ln{(\Delta)}$-approximation solution to the minimum dominating set for the graph $G$ can be found in $O(n\Delta\log(\lambda))$ rounds using Algorithm~\ref{alg: approx_DS} with $O(\log (n))$ bits of memory per robot, where $\lambda$ is the maximum length of the ID-string of the robots.  
\end{theorem}
\begin{proof}
    From Lemma~\ref{lem: approximation}, we know that Algorithm~\ref{alg: approx_DS} provides a $ln (\Delta$) approximation solution. And from Lemma~\ref{lem: execute} we know that it takes at most $O(n\Delta\log(\lambda))$ rounds to execute. Thus, the theorem.
\end{proof}





\section{Conclusion and Future Work}
In this paper, we solved the problem of finding the dominating set of a graph $G$ using mobile agents. When the agents start at a single source, we are able to find a minimal dominating set in $O(m)$ rounds. On the other hand, for the multi-source starting configuration, the robots obtained a minimal dominating set in $O(\ell\Delta\log(\lambda)+n\ell+m)$ rounds. Additionally, the dominating sets obtained by these two algorithms also serve as a maximal independent set for the graph $G$. In the last section, we described an approximation algorithm that gave us a dominating set with a size that is optimal within a factor of $\ln(\Delta)$. The approximation algorithm had a running time of  $O(n\Delta\log(\lambda))$ rounds. 

In future work, it would be interesting to investigate the lower bounds - in terms of time, memory and number of robots being used, for the dominating set problem in the same distributed model. In our paper, the dominating sets are also MIS, so, it would be interesting to see if it's possible to produce dominating sets that are also connected, i.e., connected dominating sets. It would also be exciting to explore other classical graph problems, such as ruling sets, colouring etc.,  And of course, given practical real-world concerns, adaption of the algorithms for faulty robots is an important aspect to be considered for the future as well.

\bibliographystyle{splncs04}
\bibliography{reference}





\end{document}